\documentclass[11pt,psamsfonts,reqno]{article}
\newcommand{ \hsp}{  \hspace{0.3pt}}

\usepackage{amsmath,amssymb,amsthm}    
\usepackage{graphicx}

\usepackage{hyperref}   

\usepackage{amsfonts}
\newtheorem{theorem}{Theorem}[section]

\newtheorem{lemma}[theorem]{Lemma}
\newtheorem{proposition}[theorem]{Proposition}
\newtheorem{cor}[theorem]{Corollary}
\newtheorem{fact}[theorem]{Fact}
\newtheorem{problem}[theorem]{Problem}
\theoremstyle{definition}
\newtheorem{definition}[theorem]{Definition}

\newtheorem{notation}[theorem]{Notation}
\newtheorem{remark}[theorem]{Remark}

\newcommand{\wtilde}{\widetilde}

\newcommand{\Z}{{\mathbb Z}}
\newcommand{\F}{{\mathbb F}}

\newcommand{\R}{{\mathbb R}}
\newcommand{\C}{{\mathbb C}}

\newcommand{\poly}{{\rm poly}}

\DeclareMathOperator*{\Ex}{\mathbf E}

\DeclareMathOperator{\Tr}{Tr}

\usepackage{xy}
\xyoption{matrix}
\xyoption{frame}
\xyoption{arrow}
\xyoption{arc}

\usepackage{ifpdf}
\ifpdf
\else
\PackageWarningNoLine{Qcircuit}{Qcircuit is loading in Postscript mode.  The Xy-pic options ps and dvips will be loaded.  If you wish to use other Postscript drivers for Xy-pic, you must modify the code in Qcircuit.tex}
\xyoption{ps}
\xyoption{dvips}
\fi

\entrymodifiers={!C\entrybox}

\usepackage[justification=centering]{caption}

\newcommand{\ts}{\textsuperscript}

\usepackage{color}

\usepackage{enumitem}

\usepackage{mathabx}
\usepackage{titlesec}
\usepackage{fullpage}
\titlespacing*{\section}{0pt}{7pt}{4pt}
\titlespacing*{\subsection}{0pt}{3pt}{2pt}
\titlespacing*{\paragraph}{0pt}{5pt}{3pt}

\titleformat{\section}[block]
{\LARGE\bfseries}
{\thesection.}{3pt}{}{}

\titleformat{\subsection}[block]
{\Large \sffamily}
{\thesubsection.}{3pt}{}[]

\titleformat{\paragraph}[runin]
{\large \bfseries}
{}{}{}[]

\begin{document}

\title{\LARGE \bfseries Information Causality, Szemer\'{e}di-Trotter and \\ algebraic variants of CHSH } 

\author{\sc Mohammad Bavarian\thanks{\texttt{bavarian@mit.edu}. Supported by NSF-STC Award 0939370, and NSF CCF-1065125} \\ \sc MIT  \and   \sc Peter W. Shor\thanks{\texttt{shor@mit.edu}. Supported by NSF grant CCF-0829421, and by the STC award for Science of Information under NSF grant CCF-0939370.}\\  \sc MIT }
\date{}
\maketitle
\begin{abstract}
In this paper, we consider the following family of two prover one-round games. In the $\sf{CHSH_q}$ game, two parties are given $x,y\in \mathbb{F}_q$ uniformly at random, and each must produce an output $a,b\in\mathbb{F}_q$ without communicating with the other. The players' objective is to maximize the probability that their outputs satisfy $a+b=xy$ in $\mathbb{F}_q$. This game was introduced by Buhrman and Massar \cite{buhrman} as a large alphabet generalization of the {\sf{CHSH}} game---which is one of the most well-studied two-prover games in quantum information theory, and which has a large number of applications to quantum cryptography and quantum complexity. 
Our main contributions in this paper are the first asymptotic and explicit bounds on the entangled and classical values of $\sf{CHSH_q}$, and the realization of a rather surprising connection between $\sf{CHSH_q}$ and geometric incidence theory. On the way to these results, we also  resolve a problem of Paw{\l}owski and Winter \cite{hyperbits} about pairwise independent Information Causality, which, beside being interesting on its own, gives as an application a short proof of our upper bound for the entangled value of $\sf{CHSH_q}$.

\end{abstract}

\section{Introduction}
In this work, we study a certain family of two prover one-round games. The study of multiprover one-round games (from now on, simply referred to as games) began in the late $20^{\rm{th}}$ century in the context of multiprover interactive proof systems in computer science \cite{ben1988}, and also in the context of the Bell inequalities in physics \cite{bell1964} with the topic continuing to be of significant interest in both computer science and quantum physics to this day (see for example \cite{briet2013,buhrman2011,dinur2014,junge2011,raz2011}).  The particular family of games we shall study was first introduced by Buhrman and Massar \cite{buhrman} nearly a decade ago. It is defined as follows. 

\begin{definition} \label{CHSH_q_def} Let $q$ be a prime, or a prime power, and $\F_q$ the unique field of size $q$. In the $\bf\mathsf{CHSH_q}$ game, two non-communicating parties Alice and Bob are each given an input $x$ and $y$ from $\F_q$ chosen uniformly at random. Their objective is to maximize the probability that their outputs $a,b\in \F_q$ satisfy $a+b=xy$.
\end{definition}

\begin{figure}
\centerline{
\includegraphics[width=0.7\textwidth]{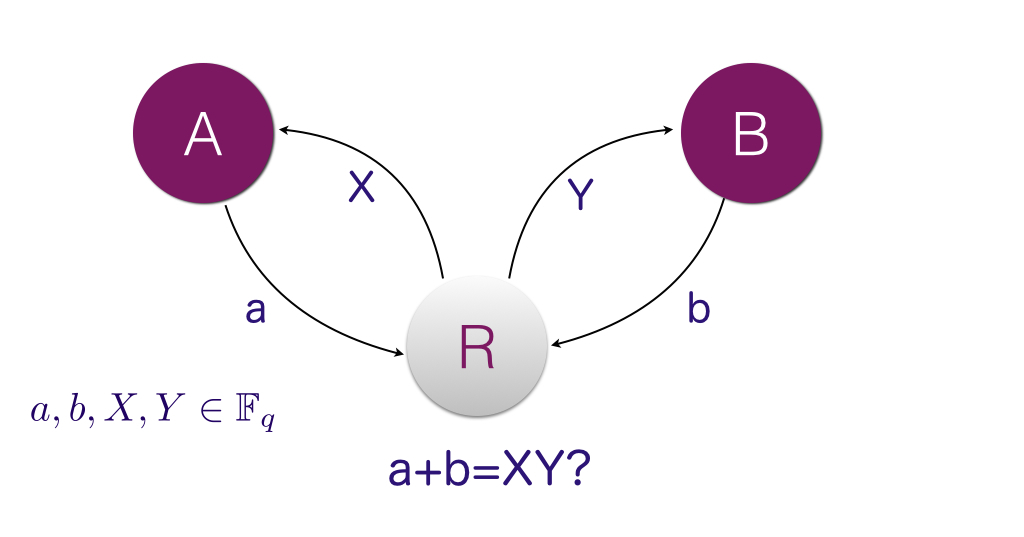}
}
\caption{The $\mathsf{CHSH_q}$ game. }
\end{figure}
The celebrated {\sf CHSH} game, named after its inventors Clauser, Horne, Shimnoy and  Holt \cite{chsh_paper}, is the case $q=2$ of the above definition. It is arguably the most well-studied game in quantum information theory \cite{werner2001bell}, and has many applications in the study of entanglement (\cite{brunner2013,werner2001bell}) and also in quantum cryptography \cite{ekert91} and quantum complexity
\cite{RUV}. Given the major role {\sf CHSH} plays in many aspects of quantum information theory, it has been of great interest to find well-structured asymptotic generalizations of {\sf CHSH}, as this could have much impact in the study of non-locality in general, and in the above applications of ${\sf CHSH}$ in particular. In this paper we focus on Buhrman and Massar's generalization, described  in Definition \ref{CHSH_q_def}, since we expect that the algebraic form of ${\sf CHSH_q}$ would lead to a interesting and useful structure for this family of games.  
In fact one of the main results of our work is the realization of a strong connection between the ${\sf CHSH_q}$ game and some remarkable mathematical results in incidence geometry and arithmetic combinatorics. This surprising connection combined with our other results further supports the intuition about the rich structure of these games. 


Despite the simple form of this family of games and our precise understanding of the case $q=2$, it turns out that analyzing $\sf{CHSH_q}$ beyond the $q=2$  case is a rather difficult task. This difficulty is not restricted to analyzing ${\sf CHSH_q}$; it is actually an instance of a more general phenomenon, and is essentially shared with any game with $q\geq 3$. The main issue here is that we do not know a  large alphabet generalization of the foundational result of Tsirelson on $\sf SDP$ characterization of the entangled value of $\sf XOR$ games (which are a subclass of games with $q=2$).\footnote{ See Definition \ref{def:xor} for a precise definition of {\sf XOR} games. The games we consider here are perhaps the most natural higher alphabet generalization of $\sf{XOR}$ games as the referee's acceptance predicate only depends on the sum over $\F_q$  of player's outputs. Remarkably, such a simple generalization  from $q=2$ to say $q=3$ seem to make a substantial difference.} This result of Tsirelson, combined with the tools of convex analysis such as complementary slackness and strong duality, gives a powerful path toward analyzing the entangled value and the optimal strategies for $\sf{XOR}$ games. The unavailability of the above powerful tool has resulted in a scarcity of results for analyzing the games in the case $q\neq 2$--- which is regarded as one of the central challenges in the study of non-local games (see \cite{brunner2013,manyquestions_fewanswers}). Indeed, a major goal of this work is to expand on the set of examples and tools available for analyzing games beyond the relatively well-understood case of $\sf{XOR}$ games, which we do in the context of studying ${\sf CHSH_q}$. We note that our results do not go  far on addressing the fundamental problems regarding the complexity of entangled two prover non-XOR games. However, we believe that for tackling this fundamental problem, a certain amount of preparatory work in the form of establishment of new tools and examples is a definite prerequisite. We hope that our work constitutes an advance in the foundation necessary for tackling the aforementioned fundamental problems.


\paragraph{Results.}
For a game $G$, we denote by  $\omega(G)$ and $\omega^*(G)$ the maximum winning probability of classical and quantum strategies, respectively. These are usually referred to as the classical and entangled values, in short. Recall that since the quantum strategies contain the classical ones as a subset, it is clear that $\omega(G)\leq \omega^*(G)$. 

Regarding the entangled value of $\bf\mathsf{CHSH_q}$, we give two different proofs of the following theorem,
which generalizes the well-known upper
bound of Tsirelson \cite{boris_1} for the original {\sf{CHSH}} game. 
\begin{theorem}\label{main_thm} For any prime or prime power $q$ we have
\[ \omega^*({\bf\mathsf{CHSH_q}}) \leq \frac{1}{q}+ \frac{q-1}{q}\frac{1}{\sqrt{q}} \, .\]
\end{theorem}
Despite advances due to several researchers \cite{buhrman, ji2008,
liang2009, wangfunc} in analyzing the value of $\bf\mathsf{CHSH_q}$ games,
prior to our work there was no result, even in conjectured form, known for the asymptotic behavior of the classical and entangled value of $\bf\mathsf{CHSH_q}$. Even for small values of $q$, most results, with the exception of the original $\frac{1}{3}+ \frac{2}{3} \frac{1}{\sqrt{3}}$ upper bound of Buhrman and Massar for $q=3$, were obtained using numerical methods. Thus, our work is the first to obtain asymptotic results on these games.  

One interesting fact about the bound in Theorem \ref{main_thm} is the striking similarity of the bound $1/q+ (q-1)/ (q\sqrt{q})$ here, with the influential tight upper bound of Tsirelson \cite{boris_1} of $1/2+1/2\sqrt{2}$ for  {\sf{CHSH}}. This striking resemblance gives rise to the natural question of asymptotic (or exact) tightness of the bound in Theorem \ref{main_thm}. 
Although we cannot answer the above questions in full, we provide some answers which clarify the situation to some extent, and highlight some of the relevant issues. 

\begin{theorem}\label{main_thm_classical} There exists a universal constant $\epsilon_0>0$ such that for any prime $p$ and $k\geq 1$ we have 
\[ 
\omega(\bf\mathsf{CHSH_q})= \left\{ 
\begin{array}{lccr} \Omega(q^{-\frac{1}{2}}) & \mbox{for} & q=p^{2k} & \hspace{-8pt}  \\
 O(q^{-\frac{1}{2}-\epsilon_0}) & \mbox{for} & q=p^{2k-1} &\hspace{-8pt} . \end{array} \right .
 \]
\end{theorem}
To prove this theorem, we adopt a new view of $\sf{CHSH_q}$. The main insight here is the following.
\begin{fact}\label{important_fact}
A classical strategy for $\bf\mathsf{CHSH_q}$ is in direct correspondence with a configuration of $q$ non-vertical lines and $q$ points in $\F_q^2$ , with no two lines having the same slope, and no two points lying on the same vertical line. Given such a configuration of lines and points, the winning probability of the corresponding strategy for $\sf{CHSH_q}$ is proportional to the number of point-line incidences.
\end{fact}
The correspondence in Fact \ref{important_fact} allows us access to some powerful results in arithmetic combinatorics where questions related to the incidences of collections of points and lines over finite fields have seen much progress recently. Most relevant to our problem is the celebrated finite field Szemer\'{e}di-Trotter theorem of Bourgain, Katz and Tao \cite{BKT} which states that, under a certain size restriction satisfied in our case, the number of incidences between a collection of points $P$ and a collection of lines $L$ is at most of the size $|P|^{\frac{3}{4}-\epsilon_0} |L|^{\frac{3}{4}-\epsilon_0}$ for some $\epsilon_0>0$.\footnote{The original proof of Bourgain \emph{et al.}~did not show an explicit bound on $\epsilon_0$ however work in recent years have obtained explicit bounds of the form $\epsilon_0>\frac{1}{700}$. However, the truth is believed to be much better than what is proved in these works.} This result combined with Fact \ref{important_fact} is essentially sufficient to prove the upper bound in Theorem \ref{main_thm_classical}. 

In fact, the relation between ${\sf CHSH_q}$ and the finite field Szemer\'{e}di-Trotter theorem is closer than it might first appear to be. As we show in Section \ref{sec:classical}, understanding the classical value of $\sf CHSH_q$ is in some sense equivalent to the finite field Szemer\'{e}di-Trotter theorem with the appropriate parameters. We show this by proving that the restrictions on the points and lines in Fact \ref{important_fact}, which is crucial in order to translate the geometric configuration to a legal $\sf CHSH_q$ strategy, can actually be relaxed without losing much in the bounds. 

Going back to the quantum and classical values of $\sf{CHSH_q}$, it is important to notice that our classical lower bound in Theorem \ref{main_thm_classical} shows that there is no asymptotic separation between the quantum and classical values of the game for $q=p^{2k}$ , while the classical upper bound of $O(q^{-1/2-\epsilon_0})$ leaves open the possibility of such a separation in the setting of $q=p^{2k-1}$. Hence the most obvious gap in our bounds is captured by the following problem. 

\begin{problem}[Open]\label{problem:open}Does there exists an infinite family of $q=p^{2k-1}$ such that $\omega^*({\sf CHSH_q})= \Omega(q^{-\frac{1}{2}})$, or some $\delta>0$ and an infinite family of $q=p^{2k-1}$ such that $\omega^*({\sf CHSH_q})=O(q^{-\frac{1}{2}-\delta})$?\hspace{2pt}\footnote{It is possible that the answer to both questions above are positive. In fact this would be the best possible outcome from the point of view of applications.}
\end{problem}

Given the geometric picture in Fact \ref{important_fact}, we observe that our main open problem, Problem \ref{problem:open}, is related to a question of Kempe and Kasher \cite{kasher2012} about the security of Bourgain's two-source extractor \cite{bourgain2005, Rao2007} in the presence of quantum memory. The main point is that Bourgain's extractor consists of two main ingredients: The first is a crucial preprocessing step, which (roughly-speaking) makes sure the two sources are in generic position with respect to each other. The second is an application of a Hadamard extractor on the two preprocessed sources. It was shown by Kasher and Kempe that the bare Hadamard extractor remains secure in the presence of quantum adversaries. Hence, the missing part in the analysis of Bourgain's extractor in the presence of quantum memory is the first step, analysis of which heavily relies on Szemer\'{e}di-Trotter theorem on finite fields. Thus, the core of both Kasher and Kempe's question and that of ours seem to be the extent to which finite-field Szemer\'{e}di-Trotter theorem can be (or fails to be) extended to the quantum setting.

\subsection{Techniques}
Let us start by giving more detail on our two different proofs of
Theorem \ref{main_thm}. A common aspect of both these methods is that they
avoid a direct analysis of the norms of associated game operators.
Instead, they take a novel indirect approach via reductions. In order to rule out a certain winning probability $p$ for $G$, we show that
 the ability to win instances of $G$ with probability greater than $p$ would allow us 
to achieve a winning probability $p'$ for a more generic game $G'$, one which we
already know to be impossible.

Both our methods for proving the
upper bound on the entangled value of $ \mathsf{CHSH_q}$ game work by a
reduction to another generic result: the first method uses a reduction to a
large alphabet variant of the result of Linden et al.~on
quantum and classical strategies for certain distributed tasks \cite{nonlocal_comp}.\footnote{Linden et al.~prove their result only about $\F_2$-output games. We generalize this result to larger alphabets only as much as we need for our particular application. A full generalization of all aspects of the result of Linden et al. could be interesting but is beyond the scope of this article.} This approach has the advantage of being self-contained and quite simple. The main idea here is to analyze a slightly different variant of $\sf CHSH_q$ game, called $\sf CHSH_q^{dist}$, in which two parties receive $(\alpha,\gamma)\in \F_q^2$ and $(\beta,\delta)\in \F_q^2$, and their objective is to produce outputs $a$ and $b$ satisfying $a+b=(\alpha+\beta)(\gamma+\delta)$. 

The other approach is by a reduction to a new form of the principle of information causality due to Paw{\l}owski and Winter \cite{hyperbits}, further generalized in this paper. A benefit of this approach is that executing the approach naturally leads us to an open problem of Paw{\l}owski and Winter \cite{hyperbits} which we resolve in Appendix \ref{IC_app}. Given the amount of attention paid to information causality  principle in the foundations of quantum physics in recent years, this result is certainly significant on its own right. What makes this result even more attractive is that using it we can give as an application a very short proof of Theorem \ref{main_thm}. 

To discuss our result about \textsl{pairwise independent information causality}, it is best to first recall the standard scenario for information causality (IC) \cite{IC_main}.
 \begin{definition}[IC] \label{IC_def} In an information causality game, Alice
 is given an input $\mathbf{X}=(X_1, X_2, \ldots, X_N)$ from a known
 distribution $\pi$, and Bob an index $b\in [N]$. After making a
 measurement on her system, Alice sends a message $\alpha\in \Sigma$ to Bob. After receiving $\alpha$ from Alice, Bob makes a measurement on his system producing an output $Z \in \Lambda$. Alice and Bob's goal is to maximize the quantity $ IC(A,B)=\sum_{i=1}^N I(X_i; Z | b=i)$.\footnote{A moment of reflection shows that the distribution of Bob's input $b\in [N]$ does not play any role here. Hence,  it can be taken to be uniform over $[N]$ for simplicity.}

 \end{definition}
The main idea behind the \emph{principle of information causality} is that assuming a certain form of independence among Alice's inputs (i.e.~$\{X_i\}_{i=1}^N$), there is a stringent limit to the amount of correlation, as quantified by $IC(A,B)$, that the two parties can create by limited communication---even given arbitrarily entanglement between Alice and Bob. Our main result about the information causality game is as follows.
 \begin{theorem}[Pairwise Independent IC] \label{IC_st} Consider an information causality game as in Definition \ref{IC_def}. Assume Alice's input $\mathbf{X}=(X_1,X_2,\ldots, X_N)$ is drawn from an unbiased (i.e.~with uniform marginals) pairwise independent distribution. Then we have 
 \[ IC(A,B)= \sum_{i=1}^N I(X_i; Z |b=i)=O_{|\Sigma|, |\Lambda|} (1)\, , \]
where $O_{|\Sigma|, |\Lambda|}(1)$ is a quantity depending only on the sizes of the alphabets of Alice's message to Bob and Bob's output, and not $N$.\footnote{In most applications of information causality type theorems, including ours, the exact dependence on $|\Sigma|$ and $|\Lambda|$ is not important. The bound here is linear in $|\Sigma| |\Lambda|$.}
 \end{theorem}
The original setting of information causality from \cite{IC_main} is the case where $\{X_i\}_{i=1}^{N}$ are fully independent. Paw{\l}owksi and Winter in \cite{hyperbits} strengthened the original information causality by showing that a similar bound holds even if the full independence condition is \textsl{relaxed} to pairwise independence, under the restriction that both $Z$ and $\alpha$ have only two outcomes. They posed as an open problem to extend their result to larger alphabets. As it turns out such a theorem is precisely what we need to prove Theorem \ref{main_thm} by our information theoretic approach. We prove this Theorem in Appendix \ref{IC_app} and use it in Section \ref{IC_approach_sec} to prove Theorem \ref{main_thm}. 
 
Finally we end this part with a technical remark about Theorem \ref{IC_st}.
\begin{remark}
(i) There is a crucial difference between $I(X_1 \ldots X_N; Z|b=i)$ and $\sum_{i=1}^N I(X_i; Z |b=i)$. It follows from the data processing inequality that the former is bounded by $O_{|\Sigma|}(1)$ with no need for any independence assumption. On the other hand, taking $X_1=X_2=\ldots= X_N={\rm Bernoulli}(\frac{1}{2})$, we see that $IC(A,B)$ could be as large as $\Omega(N)$ in this case. Thus, some independence assumption is necessary for  Theorem \ref{IC_st} type results. (ii) The original information causality result of \cite{IC_main} can be proved by an appropriate use of data processing inequality and the chain rule applied to $I(X_1 X_2 \ldots X_N; \alpha)= O_{|\Sigma|}(1)$. 
Similar approaches run into trouble in the pairwise independence setting because of the correlations between $X_i$'s. 
\end{remark}
\subsection{Prior work}\label{prev_work}
Buhrman and Massar were the first to study $\bf\mathsf{CHSH_q}$ for $q\neq 2$ obtaining the upper bound of $1/3+2/3\sqrt{3}$ for $q=3$ using information theoretic methods (different from the ones used here). Although the algebraic view of {\sf{CHSH}} is in retrospect more or less clear, this was not the language originally used to describe the {\sf{CHSH}} game.  Hence, the contribution of Buhrman and Massar was to both realize this view, and to tackle the next interesting case after the original {\sf{CHSH}}, which was the case of $q=3$. In the same work, they mentioned that their method seemed not to work for higher values of $q$. After the work of Buhrman and Massar, the problem was attacked for small values of $q$ by Ji et al. \cite{ji2008} and then by Liang, Lim and Deng \cite{liang2009} who, through a mix of numerical work and analytic insights, obtained several upper and lower bounds for quantum and classical value of the games for $q$'s up to  $13$. Since the approach in the above line of work is mostly numerical, it is hard to infer much about the asymptotic questions of interest from the bounds there. The only work prior to ours to obtain some general results about $\bf\sf CHSH_q$ is \cite{wangfunc}.  There, Wang proved various interesting results including a large alphabet generalization of a result of van Dam \cite{vandam}, on the collapse of communication complexity in the presence of perfect $\bf\sf CHSH_q$ oracle boxes (this is a natural higher alphabet analogue of the well-known Popescu-Rohlich box \cite{PR_box}). He also approached the problem of analyzing the value of $\sf{CHSH_q}$ using the principle of information causality; however, the arguments there did not seem to provide explicit bounds. We discuss the similarities and the differences between our approach and Wang's in more detail in Section \ref{sec:expos}. This provides some intuition on why basic information causality of \cite{IC_main} seems insufficient to achieve the bound in Theorem \ref{main_thm} showing the advantage offered by the strengthened version in Theorem \ref{IC_st}.

Before finishing the discussion of the prior work, let us take some time to elaborate more on the work of Liang et al.~\cite{liang2009}. The upper bounds on their work is based on numerically solving the \textsf{SDP} hierarchies from \cite{navas2008}. Their method of lower bound (described on the bottom of page 4 of their paper) is a variant of the following natural heuristic: One starts with an arbitrary choice of Alice's and Bob's strategies and a joint state $\psi \in \C^d\otimes \C^d$ in a candidate dimension $d$. Then one runs the following three-step iterative procedure: We first optimize for Bob's strategy, given the state $\psi$ and Alice's strategy. Then, we optimize Alice's strategy fixing Bob's strategy and the state $\psi$. Finally, we optimize for $\psi$ while fixing everything else. The first two steps are \textsf{SDP} computations, and the latter is simply a maximum eigenvalue computation. The above process is repeated until a (near) local optimum is reached.\footnote{
Let us note that an anonymous reviewer, based on their own personal experiments, has raised some doubts regarding the validity of the lower bounds for $q=3$ in Liang et al. We believe that in this case Liang et al.~result is correct, as prior to becoming aware of their result we independently obtained their lower bound for $q=3$ via different methods.}

\subsection{Organization of the paper} To gain a better understanding of the main ideas behind the information theoretic approach in analyzing the two prover games, we recommend the reader to review the exposition of the approach as presented in Section \ref{sec:expos}. Although the arguments in Section \ref{sec:expos} are not essential for proving the main results, it may provide much intuition and make the steps in various arguments more transparent. The most technical part of the paper is Appendix \ref{IC_app} which contains the proof of the open problem of Paw{\l}owski and Winter. This section can be skipped if one is only interested in the application of Information Causality type principle to two-prover games rather than the proof of such results. Also for the sake of brevity, the second proof of Theorem \ref{main_thm} based on reduction to distributed non-local computation and  the proof of Lemma \ref{classical:lemma}, which requires more familiarity with projective spaces, are transferred from the main body to the appendices in this version.
\section{Preliminaries}\label{prelim}
In this work, $q$ will denote a prime power unless otherwise noted. We let $[n]$ denote the set of positive integers in the range $1$ to $n$. By $O_{\alpha,\beta}(\cdot)$ we mean a quantity which is bounded by a universal constant depending only on the parameters $\alpha$ and $\beta$. In parts of the paper, we require a basic knowledge of information theory. A brief introduction to relevant  notions from this area is presented in Appendix \ref{appendix:info_theory}.

By a game we always mean a two prover one-round refereed game unless otherwise specified. For our purposes, we only need the definition of following subclass of games.
\begin{definition}\label{def:xor} An {\sf XOR}-game $G$ is specified by the following data:  a set of question $Q_1$ for Alice and a set of question $Q_2$ for Bob, a distribution $\pi$ over $Q_1\times Q_2$, and matrix $V$ with entries from $\F_2$ of size $Q_1\times Q_2$. Given $q_1\in Q_1$ to Alice, and $q_2 \in Q_2$ to Bob with probability $\pi(q_1,q_2)$, Alice and Bob must each produce $a,b\in \F_2$ respectively. They succeed if $a+b= V(q_1,q_2)$. If in the above setting one replaces the additive group of $\F_2$ with the additive group of $\F_q$, the resulting class of games will be called {\sf q-XOR} games.
\end{definition}
 For \textsf{q-XOR} games, it is convenient to define the \emph{bias} of a strategy as an alternative way of quantifying the winning probability. 
\begin{definition}\label{reg_def}
For any probability $0\leq p_{win}\leq 1$ corresponding to winning probability of some strategy for a \textsf{q-XOR} game, we define the bias parameter of the strategy to be 
\[E \overset{\rm{def}}{=} (qp_{win}-1)/(q-1)\, .\] 
\end{definition}
With this definition we have $ p_{win}= \frac{1}{q}+ \frac{q-1}{q}E$. Note the bias parameter quantifies by how much a strategy outperforms the trivial random strategy which for \textsf{q-XOR} games achieves the expected winning probability of $1/q$. Hence, we  have $E\geq 0$ for the optimal strategy of any \textsf{q-XOR} game.

\begin{remark}
We note that the notion of bias can also be defined in the same way in the more general setting of unique games with alphabet size $q$. However, we shall not need to work in this general setting for the purposes of this work and we do not require the definition of what a  unique game is. However, the reader familiar with the terminology should note that all \textsf{q-XOR} games, and in particular $\mathsf{CHSH_q}$ games, are unique games.
\end{remark}

\section[Information Causality Approach] {Overview of Information Causality approach}\label{sec:expos}
Here we discuss the basic method that allows one to use information causality type theorems to prove bounds on the value of $\sf{CHSH_q}$ game. Since more precise arguments with better bounds are provided in Section \ref{sec:upperbound} and Appendix \ref{app:nonlocal}, we omit some of the details. For simplicity, we work in the fully independent setting of the information causality game; this has the added benefit of clarifying why we need the strengthened version of the information causality principle, Theorem \ref{IC_st}, in order to prove Theorem \ref{main_thm}. In this section, for the ease of notation, we use a slightly different convention from other sections by taking Alice's input to be indexed from $0$ to $N-1$, as opposed to $1$ to $N$. 

The idea is best captured in the noiseless setting: assume that we have a (quantum or whatever) strategy that allows us to win the $\sf CHSH_q$ game with probability $1$. Let us see that this leads to some unlikely consequences.
\begin{figure}[h]
\centerline{
\includegraphics[width=0.48\textwidth]{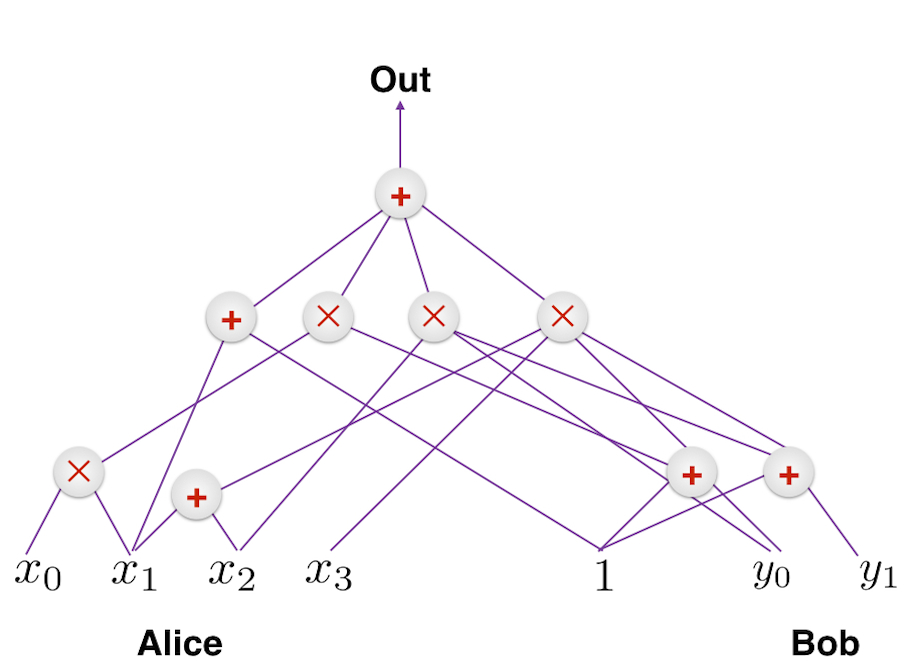}
\hspace{1pt}
\includegraphics[width=0.48\textwidth]{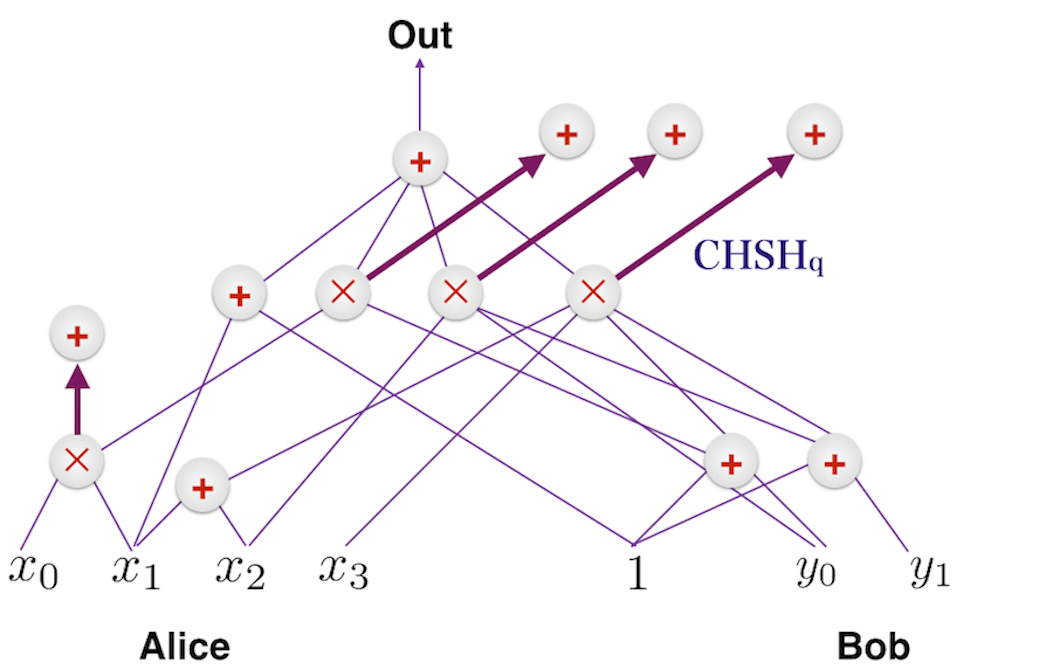}
}
\caption{\small The players' strategies for  $\mathsf{CHSH_q}$ can be used to \emph{eliminate} $\times$ gates from a distributed circuit, replacing them with $+$ gates---at the cost of some noise.}
\end{figure}

Let $q=2$ for simplicity. Assume Alice is given $X_0, X_1, X_2, X_3\in \F_2$ uniformly at random. Bob is given an index pointing to one of the Alice's inputs in the form of $b=b_0b_1$ with $b_i\in \{0,1\}$. The main observation is that the desired output of Bob, which is $X_b$, can be written as a polynomial over $\F_2$ in terms of $X_i$'s and $b_j$'s. Let us denote the desired output of Bob by $g_2(\cdot)= X_b$. In general $g_k$ denote the polynomial that Alice and Bob want to compute in the ``k\ts{th} level"  of the game, i.e., the setting where we have $N=2^k-1$ and $b=b_0b_1\ldots b_{k-1}$. Observe that
\[  g_2(X_0, X_1, X_2, X_3, b_0, b_1)= b_0 b_1 X_3 + (1-b_0) b_1 X_2 + b_0 (1-b_1) X_1 + (1-b_0) (1-b_1) X_0\, . \]
From the above, it is clear what Alice and Bob should do: given (say) the first monomial, $b_0b_1 X_3$, Alice and Bob could use their perfect $\sf CHSH$ strategy to produce $u_3+v_3= b_0b_1 X_3$. Doing the same thing for all other monomials, they essentially dispose of all multiplications and are left with
\[ g_2(X_0, X_1, X_2, X_3, b_0, b_1)=\sum_{i=0}^3 u_i + \sum_{i=0}^3 v_i \, .\]
Hence, if Alice sends $\sum_{i=0}^3 u_i$ to Bob, using just one bit of communication, Alice and Bob manage to compute $g_2(\cdot)$ successfully. It is clear that the above perfect strategy can be extended in the same manner to arbitrary ``level $k$," i.e. when Alice is given $\{X_i\}_{i=0}^{2^{k}-1}$ and Bob an index $b=b_0b_1\ldots b_{k-1}$, with no need for more than $1$-bit of communication. Hence, it seems that there must be some information theoretic impossibility here,  as one bit of communication should not be sufficient to allow two players to compute the rather complicated address function $g_k(\cdot)$ for large $k$. This is exactly the main idea of the information theoretic approach for Theorem \ref{main_thm}.

To prove Theorem \ref{main_thm} we need a more careful version of the above argument in order to derive a contradiction even when Alice and Bob only share a \emph{noisy strategy} for $\sf{CHSH}$ or $\sf{CHSH_q}$. Let us give some more detail: imagine Alice and Bob now only have a strategy that allows them to win with probability $w<1$. Instead of being as generous as before, Alice and Bob should now try to minimize the number of times they use their strategy for $\sf{CHSH}$ to transform a multiplication into an addition. Naively, if there exists a circuit with $m$ multiplication gates for computing the polynomial $g_k(\cdot)$, one would expect the winning probability of Alice and Bob to scale like $w^m$. As one increases $k$, if Alice's and Bob's probability of success in computing $g_k(\cdot)$ deteriorates slowly enough compared to the rate of increase in the information complexity of computing $g_k(\cdot)$, this gives us our desired contradiction.

The above illustrates the basic strategy for proving upper bounds on winning probabilities of games, in our case $\sf CHSH_q$, using information theoretic methods. The above analysis can be improved in several ways; however even with all improvements we are unable to prove Theorem \ref{main_thm} using just basic information causality result of \cite{IC_main}. The main difficulty here stems from the situation at the first level, and this difficulty in our view is one of the reasons why the method of Wang~\cite{wangfunc} and our own initial method based on the basic Information Causality theorem~\cite{IC_main} seemed incapable of reproducing the bounds in Theorem \ref{main_thm}.\footnote{The other major difference between our work and that of Wang in our view is the use of \textsl{Regularization Lemma} \ref{reg_lem} which greatly simplifies our analysis by reducing the number of parameters we have to keep track of in Proposition \ref{prop_easy} to the bare minimum.}  To demonstrate this, assume $q=3$. Let $X_0,X_1, X_2,b_0\in \F_3$ and let $h_1(\cdot)$ be the polynomial over $\F_3$ which Alice and Bob would like to compute at level $1$. We have

\begin{align*} h_1(X_0,X_1, X_2, b_0) &= (1-b_0^2) X_0 + (1 - (b_0-1)^2) X_1 + (1- (b_0 +1)^2) X_2 \\
&= X_0+ (X_2-X_1) b_0 - (X_0+X_1+X_2) b_0^2 \, . \label{rr}
\end{align*}
We see that the number of multiplications we have at level $1$ for $\F_3$ is $2$; on the other hand, over $\F_2$ we have $g_1(X_0, X_2, b_0)= X_0 + (X_0+X_1)b_0$. The increase in the number of multiplication gates per level is the main problem which forces us to use Theorem \ref{IC_st} as opposed to the fully independent version. To see the advantage of Theorem \ref{IC_st}, suppose that we choose Alice's input uniformly at random from the subspace $X_0+X_1+X_2=0$. Notice that although $X_i$'s are now correlated, they are still pairwise independent.  The main advantage of taking $X_i$'s from subspace $X_0+X_1+X_2=0$ is that this allows to save one multiplication in the evaluation of $h_1$ since 
\[ h_1(X_0,X_1,X_2,b_0)= X_0 +(X_2-X_1) b_0 \, \]
in this case. The above method indeed can be extended to higher values of $k$ to prove Theorem \ref{IC_st}. Essentially, the idea is to take Alice's input uniformly at random from an appropriately chosen subspace of the vector space (over $\F_3$) spanned by $\{X_i\}_{i=0}^{3^k-1}$. However, using the full strength of Theorem \ref{IC_st}  we can actually give a simpler information theoretic proof of Theorem \ref{main_thm}, as presented in Section \ref{IC_approach_sec}.

\section[Generalized Tsirelson bounds]{Generalized Tsirelson bound for $\bf\mathsf{CHSH_q}$}\label{sec:upperbound}
The aim of this section is to prove Theorem \ref{main_thm}. We have two different proofs of this result. Here we just present the proof based on the pairwise independent Information Causality principle, Theorem \ref{IC_st}. The second proof, which is based on a large alphabet generalization of result of Linden et al.~\cite{nonlocal_comp}, is presented in Appendix \ref{app:nonlocal}. 
In both proofs of Theorem \ref{main_thm}, we can substantially simplify the arguments by assuming the optimal strategy $\mathcal P$ for $\bf\mathsf{CHSH_q}$ always produces different types of error with equal probability. This is formalized as follows. 
\begin{definition}
A classical or quantum strategy for $ \bf{\sf{CHSH_q}}$ is called \emph{regular} if the following holds:
\[ \Pr_{a,b\leftarrow\mathcal{P^*}}[ a+b= xy+ k|x,y]= \frac{1}{q}- \frac{E}{q}\qquad \forall\, k\in \F_q^* .\]
Here, $a,b$ are the outputs of Alice and Bob's strategies given $x,y$ as inputs, respectively. The symbol $a,b\leftarrow \mathcal P$ means that the outputs $a,b$ of Alice and Bob's strategies are produced via the protocol $\mathcal P$ (given $x,y$ as inputs). The symbol $E$ as usual denotes the bias of the game defined as in Definition \ref{reg_def}. In some occasions, for the ease of notation, we do not write out $a,b\leftarrow \mathcal P$ fully, and instead simply write $\mathcal P$ to represent the fact that the players follow the particular strategy $\mathcal P$ for producing their outputs.
\end{definition}

It turns out that we can without loss of generality assume any protocol $\mathcal P$ for $\bf \sf CHSH_q$ is regular which is the content of the following lemma.
 \begin{lemma}[Regularization Lemma] \label{reg_lem}
Given any protocol $\mathcal P$ for $\bf \sf{CHSH_q}$, there exists a generic method to obtain a regular protocol $\mathcal P^*$ from $\mathcal P$ without changing the winning probability. 
\end{lemma}
\begin{proof}
Given any quantum strategy $\mathcal P$ for $\bf\mathsf{CHSH _q}$ game, define its regularized version $\mathcal P^*$ as follows: On inputs $x$ and $y$,  $A$ and $B$ use shared randomness to agree upon $\alpha,\beta\in \F_q^*$ and $\gamma,\delta\in \F_q$ uniformly at random. Then, they follow the original strategy on inputs $\wtilde{x}=(\alpha x+\gamma)$, $\wtilde{y}=(\beta y+\delta)$. Let $\wtilde{a}$ and $\wtilde{b}$ be their outputs here.  Finally, $A$ outputs $a=\frac{1}{\alpha\beta}(\wtilde{a}-\delta \alpha x-\gamma\delta)$ and $B$ outputs $\frac{1}{\alpha\beta}(\wtilde{b}-\beta\gamma y)$. We show $\mathcal P^*$ satisfies the properties we wanted.

First, notice that $\mathcal P^*$ has the same winning probability as $\mathcal P$ since the input distribution remains uniform on $\F_q\times \F_q$, i.e. $\Pr_{\mathcal P^*}[a+b=xy|x,y]= \Pr_{\mathcal P, x',y'}[a'+b'=x'y'] =p_{win}\,.$ 
It is not too hard to see that
\[ \Pr_{a,b \leftarrow\mathcal P^*}[a+b=xy+s|x,y]=\Pr_{a',b'\leftarrow\mathcal P, x',y'} [a'+b'=x'y'+s\, \alpha\beta ] \, . \]
The key here is that quadruple $(x',y',\alpha,\beta)$ have a uniform product distribution over its domain $(\F_q)^2\times(\F_q^*)^2$. For $s\in \F_q^*$, we have that $s\alpha\beta$ is uniformly distributed over $\F_q^*$. This and the fact that the input $(x',y')$ is uniform on $\F_q\times \F_q$ imply the regularity property.
\end{proof}

\subsection{Reduction to pairwise independent information causality}\label{IC_approach_sec}
Let $m$ be a positive integer which will be a parameter taken to be sufficiently large in our proof. We want to instantiate Theorem \ref{IC_st} by a subcode of generalized Hadamard code over $\F_q$. Let us recall the definition of this code:
\begin{definition} The generalized Hadamard code corresponds to an $m$-dimensional subspace of a $q^{m}-1$ dimensional vector space over $\F_q$. Given a seed $\mathbf{Y}=(Y_1,Y_2,\ldots, Y_m)\in \F_q^m$ we have a coordinate per each $\xi \in \F^{m}_q \setminus \{0\}$  defined by 
\[ {\rm{Had}}_{\xi}= \xi_1 Y_1 + \xi_2 Y_2 + \ldots + \xi_m Y_m \, .\]
Hence, a codeword is a point in the above subspace and it is given by the list of $q^{m}-1$ coordinates defined as above. 
\end{definition}
The overall plan is to let Alice's input be a random codeword from $\{ {\rm Had}_\xi \}_{\xi \in \F_q^m \setminus \{0\} }$, and Bob's input to be some $\xi \in \F_q^m \setminus \{0\}$ which is an index to one of the coordinates of Alice's input. This, however, does not quite work as the generalized Hadamard codeword is not pairwise independent. To fix this issue, we instead take Alice's input to be a proper subset $U_m\subseteq \F_q^m \setminus \{0\}$ such that $\{ \rm Had_\xi \}_{\xi\in U_m}$ is pairwise independent. Concretely, we take $U_m$ to consist of $a \in \F_q^m\setminus \{0\}$ with their first non-zero coordinate equal to $1$. \footnote{This is chosen such that the map $\pi: (\xi_1, \xi_2, \ldots, \xi_m) \mapsto (\xi_1: \xi_2 : \ldots: \xi_m)$ injects onto $\mathbb{P}\F_q^{m-1}$} With this setup we have $n=\frac{q^m-1}{q-1}$, Bob's input is some $b\in U_m$ and Alice's input is a uniformly random word from $\{ {\rm Had}_\xi \}_{\xi \in U_m}$. 

Now we need a proposition.

\begin{proposition}\label{prop_easy} Let $\mathcal P$ be a regular protocol for $\mathsf{CHSH_q}$ with bias $E$. Assume Alice is given $(c_1, c_2, \ldots, c_m)\in \F_q^m$ and Bob $(d_1, d_2, \ldots, d_m)\in \F_q^m$. Assume Alice and Bob use the protocol $\mathcal P$ once per input pair $\{(c_k, d_k)\}_{k=1}^m$ to produce $\{a_k\}_{k=1}^m$ and $\{b_k\}_{k=1}^m$, i.e. using $(a_k,b_k)\leftarrow \mathcal P(c_k,d_k)$. Let $Z=\sum_{i=1}^m a_k + \sum_{i=1}^m b_k$. We have
\[ \Pr\left[ Z= \sum_{i=1}^m c_i d_i\right]= \frac{1}{q} + \frac{q-1}{q} E^m \, .\]
Also, for all $e\in \F_q^*$ we have
\[ \Pr\left[ Z= \sum_{i=1}^m c_i d_i + e\right]= \frac{1}{q} - \frac{1}{q} E^m\, . \]
\end{proposition}
\begin{proof}
The proof is by induction on $m$. For $m=1$ this is clear. Assume the result for $m-1$. Notice that for $Z=\sum_{i=1}^m c_i d_i$ to occur, it must be the case that the error in level $m-1$, i.e. $\sum_{i=1}^{m-1} c_i d_i-a_i-b_i$, exactly cancels out the error occurred in the last step, which is $a_m+b_m- c_m d_m$.  Now by the assumptions,
\begin{align*}
 Pr\left[ Z= \sum_{i=1}^m c_i d_i\right] &= \left(\frac{1}{q} + \frac{q-1}{q} E^{m-1}\right) \left(\frac{1}{q} + \frac{q-1}{q} E \right) \\
 &+(q-1) \left(\frac{1}{q} -\frac{1}{q} E^{m-1}\right)  \left(\frac{1}{q} - \frac{1}{q} E\right) \\
& = \frac{1}{q} + \frac{q-1}{q} E^m\, .
\end{align*}
This finishes the first claim of the theorem. The second claim of the theorem follows from the first one and the symmetry. 
\end{proof}
We apply this proposition to the setting where Alice is given $\{ {\rm Had}_\xi \}_{\xi \in U_m }$ and Bob $\xi^* \in U_m$. Alice and Bob want to compute ${\rm Had}_{\xi^*}= \xi^*_1 Y_1 + \xi^*_2 Y_2 +\ldots + \xi^*_m Y_m$. Here, $Y_i$'s are given to Alice as part of her inputs as any coordinate of the form $(0,0, \ldots, 1, 0, \ldots, 0)$ is in $U_m$ and $\xi^*$ is precisely the input to Bob. Hence by following the above protocol, Alice ends up sending $\sum_{i=1}^m a_i$ to Bob as the message and Bob outputs $Z=\sum_{i=1}^m (a_i+b_i)$. The resulting output $Z$ will satisfy the following:
\begin{align*}
I(X_{\xi};Z |b=\xi) &= (\frac{1}{q} +\frac{q-1}{q}E^m) \log_2(1+ (q-1)E^m)\\
&+  \frac{q-1}{q^2} (1-E^m) \log_2( 1- E^m) \, .
\end{align*}
Notice that the calculation above was not too hard because the regularity guaranteed from Proposition \ref{prop_easy} specifies the exact joint distribution of $(X_\xi,Z)$. Now we take $m$ large enough such that $E^m \ll \frac{1}{q}$. In this regime, it is easy to see that above expression is always larger than $ \frac{E^{2m}}{\poly(q)}$ for a fixed polynomial independent of $m$. Using  $I(X_{\xi};Z |b=\xi)  \geq \frac{E^{2m}}{\poly(q)}$
in Theorem \ref{IC_st}, and noticing that $|U_m|= (q^m-1)/(q-1)$, we see that $\frac{q^m-1}{q-1}\frac{E^{2m}}{\poly(q)}= O_q(1) \, .$
For this to hold for arbitrarily large $m$, we must have
\[E\leq \frac{1}{\sqrt{q}},\]
which is our desired result. 

\section[Classical Aspects]{Classical aspects of $\bf\mathsf{CHSH_q}$ and point-line incidences }\label{sec:classical}
In this section, we present our results regarding the classical value of $\sf{CHSH_q}$. This includes Theorem \ref{main_thm_classical} and various other results. 

We begin by a short introduction to some notions from geometric incidence theory. Let $\sc{\Pi}=\F^2$ be the plane over a field $\F$. For a collection of lines $L$ and points $P$ over a plane $\sc{\Pi}$ we define the set of incidences as 
\[ I(P,L)= \{ (p, l) \in P\times L \, , \, p\in l \} \, . \]
A central question in geometric incidence theory is the following: Given $|P|$ and $|L|$, what can be said about the size $| I(P,L)|$? If $|P|$ and $|L|$ are of roughly the same size, it is hard to imagine a configuration where every line in $L$ would contain every point in $P$. Hence $|I(P,L)|\leq |P| |L|$ seems a rather pessimistic upper bound. In fact, using the fact that there is at most one line through two distinct points and one point at the intersection of two distinct lines suffices to get a better upper bound of $|P|^{3/4}|L|^{3/4}+ |P| +|L|$.  As shown by Szemer\'{e}di and Trotter \cite{ST1983} the above bound can be improved when  $\F=\R$ to $|P|^{2/3} |L|^{2/3} + |P|+|L|$. The proof of this result is very geometric and relies on localization techniques that do not work in the finite field settings. The situation over finite fields remained unclear until about a decade ago; finally Bourgain, Katz and Tao \cite{BKT} used tools from arithmetic combinatorics to show an improved upper bound on $|I(P,L)|$ as long as the sets $P$ and $L$ are not too large. More specifically we have:
\begin{theorem}[BKT] \label{BKT}
For any $\delta>0$, there exists some $\epsilon>0$ such that for any prime field $\F_p$ and any collection of points $P$ and lines $L$ over $\F_p^2$ satisfying $p^\delta \leq |P|, |L| \leq p^{2-\delta}$, we have 
\[ | I(P, L)|= O\left( |P|^{\frac{3}{4}-\epsilon}|L|^{\frac{3}{4}-\epsilon} \right)\, . \]
\end{theorem}
Although the tƒheorem of Bourgain et al.~as stated above only holds for prime fields, it is not too hard to see that essentially the same argument goes through whenever $|P|, |L|$ are large compared to the proper subfields of $\F_q$. This is made explicit in the work of Jones \cite{Jones_2010, Jones_thesis}. Although what is proved in \cite{Jones_2010, Jones_thesis} is more general, we just need the following corollary.

\begin{cor}[Jones]\label{jones_cor}
Let $q$ be an odd power of a prime, and assume $P$ and $L$ are sets of points and lines in $\F_q^2$ of size $\Theta(q)$. There exists a universal constant $\epsilon>0$ such that 
\[   I(P,L) \leq q^{\frac{3}{2}-\epsilon} \, .\]
\end{cor}
\subsection{From point-line incidences to $\sf{CHSH_q}$}

Given the above results, Theorem \ref{main_thm_classical} can be proved rather quickly. To see why, recall Fact \ref{important_fact} from the introduction where it was claimed that the winning probability of any classical strategy for $\sf CHSH_q$ corresponds to the number of point-line incidences among $q$ lines and $q$ points in $\F_q^2$ under some restrictions on the points and lines.  
Hence, the result of Jones immediately implies the lower bound in Theorem \ref{main_thm_classical}. For the upper bound, we need the following lemma which allows to relax the restrictions on the points and lines in Fact \ref{important_fact}. 
 
\begin{lemma}\label{classical:lemma} Let $P, L$ be a set of points and lines in $\F_q^2$ with $|P|= \Theta(q)$ and $|L|=\Theta(q)$. There exists a set of points $P'$, and a set of lines $L'$, satisfying the conditions of Fact \ref{important_fact} with $|P'|\leq |P|$ and $|L'|\leq |L|$ such that  \[ |I(P', L')|= \Omega \left( | I(P,L)| \right)\, . \] 
\end{lemma}
The main idea for proving this lemma is to start from the given configuration of $P$ and $L$, and apply a random projective transformation to them. The next step is to remove the lines with the same slope, and the points on the same vertical line. This ensures that the remaining sets of points and lines, $P'$ and $L'$, satisfy the required condition. It is not too hard to show that this deletion process shrinks the number of incidences only by a constant factor in expectation finishing the proof of Lemma \ref{classical:lemma}. Notice that one could not simply use affine transformations in the place of the projective ones as affine transformations preserve the direction of lines which is problematic if the initial configuration of $L$ has many parallel lines. A more detailed proof is given in Appendix \ref{lowerbound_classical}. 

Before proving Theorem \ref{main_thm_classical} we need to establish some useful notation.
\begin{definition}\label{line_def}
Let $\F_q^2= \F_q[z_1,z_2]$ be the plane. We denote by $\ell_{a,b}$ the line
\[ \ell_{a,b}=\{ (z_1,z_2)\in \F_q^2 : \, z_2=a\hspace{0.2pt}z_1 - b \} . \]
\end{definition}

\begin{proof}[Proof of Theorem \ref{main_thm_classical} and Fact \ref{important_fact}]
The optimal classical strategies for $\bf\mathsf{CHSH_q}$ are given by two functions $f,g:\F_q\rightarrow \F_q$ corresponding to Alice and Bob's strategies maximizing 
\[  \left| \{ x,y \in \F_q \, | \, f(x)+g(y)=xy \} \right | \, . \]
Given $f:\F_q\rightarrow \F_q$ corresponding to Alice's strategy, let $P$ be the collection of $q$  points of the form $(x,f(x))\in \F_q^2$. To Bob's strategy $g:\F_q\rightarrow \F_q$, we associate a collection $L$ of $q$ lines  $\{\ell_{y, g(y)}\}$. Observe in this language any pair $(x,y)\in \F_q^2$ satisfying $\bf{\mathsf{CHSH_q}}$ correspond to a point-line incidence. Hence, we have
\[ f(x) + g(y) = xy \quad \Leftrightarrow \quad (x, f(x)) \in \ell_{y, g(y)}\, ,  \]
which means that
\[ \sum_{x,y \in \F_q} 1_{f(x)+ g(y)= xy} = I(P,L). \]
The upper bound for $q=p^{2k-1}$ now follows from Corollary \ref{jones_cor}. \\
Now assume $q=p^{2k}$. Recall that in this case there exists a subfield $K\cong \F_{\sqrt{q}}$ such that $K\subset \F_q$. Let  $P= \{ (a,b)\in \F_q^2 \, : \; a,b\in K\}$ and $L= \{ \ell_{c,d} \, : \; c,d\in K\}$. Notice that $|P|=|L|=q$, and $|I(P,L)|=q^{3/2}$. Combined with Lemma \ref{classical:lemma}, this proves the lower bound.
\end{proof}

This framework can also be used to give an improved lower bound for general $q$'s. 

\begin{theorem}\label{lowerbound_general} There exists a strategy for $\bf\mathsf{CHSH_q}$ achieving a winning probability $\Omega(q^{-2/3})$.
\end{theorem} 
Let us note that the above lower bound was somewhat counterintuitive to us at first. The point is that we expected that the function $(x,y) \mapsto xy$ to be in some sense \textsl{maximally psuedorandom} against the function $(a,b)\mapsto a+b$. Given this, it was reasonable to assume that the best classical strategy for $\sf{CHSH_q}$ would achieve a winning probability of $\wtilde{O}(q^{-1})$ which is up to polylogarithmic factor the same as that of a random strategy. The logarithmic advantages can be seen to be achievable using a simple balls-and-bins analysis by taking a random function $f:\F_q\rightarrow \F_q$ as Alice's strategy, and optimizing Bob's strategy $g:\F_q\rightarrow \F_q$ given that of Alice. In fact, numerical experiments which looked for locally optimal solutions confirmed the above intuition.\footnote{More precisely, the algorithm used was the following: we start from two random strategies $f,g:\F_q\rightarrow \F_q$, and in every iteration we fix one of the functions and update the other one to the optimal strategy taking the other function as fixed.} Despite all this, Theorem \ref{lowerbound_general} states that much better lower bounds are achievable in general. 

\begin{proof}[Proof of Theorem \ref{lowerbound_general}]
This follows from Lemma  \ref{classical:lemma} applied to the next proposition.
\end{proof}

\begin{proposition}\label{inc_propos}
For any finite field field $\F_q$, there exists a set of at most $q$ lines and at most $q$ points over $\F_q^2$ with $\Omega(q^{4/3})$ incidences.
\end{proposition}
\begin{proof}
Let $q=p^s$. First we handle the case $s=1$, then we give a construction for $s\geq 2$. 
Since we are concerned with an asymptotic statement, we can assume $q$ is sufficiently large. As a result, we can safely ignore the ceiling and floor signs as they do not affect the asymptotic. 
For the prime case $s=1$, the construction is very simple: let $P=\left [q^{1/3} \right] \times \left [ q^{2/3}\right]$. Let $L$ be the collection of lines $\ell_{c,d}$ of the form $y=cx+d$ with $c\in [q^{1/3}/2]$ and $d\in [q^{2/3}/2]$. It is clear that this achieves $I(P,L)= \Theta(q^{4/3})$ with $|P|, |L|\leq q$. 

For the case $s\geq 2$, we choose our set $P$ to be a product set, $P=A\times B$ where $A,B\subset \F_q$ are both subspaces. Let $g$ be a primitive element of $\F_q$ so $\{1,g,g^2, \ldots, g^{s-1}\}$ form a basis of $\F_q$ as a vector space over $\F_p$.  
 
 Let $b\leq s$ be a positive integer, close to $2s/3$, to be specified later. Let $a= s - b$ (the condition $s\geq 2$ will turn out to be sufficient for $a\geq 1$ which we require). Define 
 \[ A= \F_p + g\, \F_p  + g^2 \F_p +\ldots + g^{a-1}  \F_p \]
 and 
 \[ B= \F_p + g \, \F_p + g^2 \F_p + \ldots + g^{b-1}  \F_p \]
 and 
 \[ C=\F_p + g\, \F_p + g^2\F_p +\ldots + g^{b-a} \F_p \, .\]
Notice that $|P|= |A| |B|= q$. 

Let $\ell_{c,d}\subset \F_q^2$ to be the line corresponding to $\{ (x,y)\in \F_q^2 :\, y=cx+d\}.$ Define 
\[ L= \{ \ell_{c,d}: \, c\in C , d \in B \} \, .\]
Given this we can see 
\[ I(P,L)= |A|\hsp |B|\hsp  |C|\quad, \quad |L|= |B| \hsp |C|\, . \]
We want $|L|= O(q)$ while $I(P,L)= \Omega(q^{4/3})$. Since $|A||B|=q$, it suffices to choose $b$ such that
\[ |C|=p^{2b-s+1}= \Omega( p^{\frac{s}{3}}) \;\; , \:\; |B| \hsp |C|= p^{3b-s+1} = O(p^s) \, .\]
Now if $s\; \rm{mod}\; 3=2$ then  
\[ \Z \cap \left[\frac{2s}{3} -\frac{1}{2}, \frac{2s}{3} -\frac{1}{3}\right]\neq \emptyset . \]
Hence, we are done by taking $b$ to be the integer in that interval. For,  $s=3k$, $s=3k+1$ we take $b=2k$ and $b=2k+1$ respectively. In $s=3k$ case, we have $|I(P,L)=p^{4k+1}$ and $|L|=p^{3k+1}$. The important thing is that although $|L|$ is larger than its desired size by a factor of $p$,  we are also exceeding the desired $|I(P,L)|$ lower bound by a factor $p$. A moment of though reveals that choosing $L'$ to be the subset of $L$ of size $p^{3k}$ with maximum number incidences will finish the proof in this case. The situation in  $s=3k+2$ case is analogous: if we choose $L'$ be the the subset of $L$ of size $p^{3k+1}$ with the maximum number of incident points from $P$ that will finish the construction.
\end{proof}

\section{Concluding remarks}
In this work, we initiated the study of $\bf\mathsf{CHSH_q}$ in the asymptotic setting. We developed the theory of both quantum and classical values of this family of games, and outlined the connection to the problem of point-line incidences over finite fields. The fact that $\sf CHSH_q$ is  a natural problem to consider in the study of non-{\sf XOR} games (which is the original motivation of our work as well as Burhman and Massar's) while exhibiting intimate connections to above mathematical topics indicates that this problem deserves further investigation in the future. This is especially boosted by the fact that guaranteed progress can be made by using better numerical methods to investigate higher values of $q$, and also by attempting to quantize the results in arithmetic combinatorics. An investigation of the extent to which the results in additive and arithmetic combinatorics quantize could certainly have much further impact beyond the problems considered here. 

One goal of our study was to further develop the techniques available for analyzing the entangled value of non-binary non-local games. We believe that by giving two rather different proofs of Theorem \ref{main_thm}, we demonstrated the power of the indirect approach of analyzing non-local games. 
\paragraph{Future directions.} As discussed previously, Problem \ref{problem:open} remains the most clear open problem given the bounds proved here. As mentioned before, its resolutions is likely to also resolve to  Kasher and Kempe's problem regarding the security of Bourgain's two-source extractor in the presence of entanglement\cite{kasher2012}. We can think of two possible routes for resolving this problem: one is by trying to quantize the arguments in the paper of Bourgain, Katz and Tao \cite{BKT}, and the other is by investigating the {\sf SDP} hierarchies of Navascu{\'e}s et al. \cite{navas2008} to see whether they could lead to any improvement to Theorem \ref{main_thm} or lead to tightness results via some rounding scheme. In the hierarchy approach, it might be useful to keep in mind the rounding scheme of Kempe et al. \cite{kempe2008} (though their result seem more relevant when the game value is close to $1$ which is not the case here). Currently, with some collaborators, we are pursuing the latter direction via the \textsf{SDP} hierarchies.


We finish by recounting the perhaps most intriguing (and rather open-ended) future direction. This is the question of the extent to which the relatively well-understood theory of $\sf XOR$ games extends to larger alphabets. A related question is to find a better explanation for the absence of any analogue of a large alphabet generalization of Tsirelson's theorem \cite{boris_2} for even slight variants of non-{\sf XOR} games (say {\sf q-XOR} games for $q=3$) in the literature. A better understanding of the above issues would certainly constitute a major advance in our understanding of two prover games and non-locality in general.
\bibliography{QIReferences}

\begin{thebibliography}{10}

\bibitem{bell1964}
J.~S. Bell.
\newblock On the {Einstein-Podolsky-Rosen} paradox.
\newblock {\em Physics}, 1(3), 1964.

\bibitem{ben1988}
M.~Ben-Or, S.~Goldwasser, J.~Kilian, and A.~Wigderson.
\newblock Multi-prover interactive proofs: How to remove intractability
  assumptions.
\newblock In {\em Proceedings of the twentieth annual ACM symposium on Theory
  of computing(STOC)}, 1988.

\bibitem{bourgain2005}
J.~Bourgain.
\newblock More on the sum-product phenomenon in prime fields and its
  applications.
\newblock {\em International Journal of Number Theory}, 1(01):1--32, 2005.

\bibitem{BKT}
J.~Bourgain, N.~Katz, and T.~Tao.
\newblock A sum-product estimate in finite fields, and applications.
\newblock {\em Geometric \& Functional Analysis}, 14(1), 2004.

\bibitem{briet2013}
J.~Bri\"{e}t and T.~Vidick.
\newblock Explicit lower and upper bounds on the entangled value of multiplayer
  {XOR} games.
\newblock {\em Communications in Mathematical Physics}, 321(1):181--207, 2013.

\bibitem{brunner2013}
N.~Brunner, D.~Cavalcanti, S.~Pironio, V.~Scarani, and S.~Wehner.
\newblock {Bell} nonlocality.
\newblock {\em quant-ph:1303.2849}, 2013.

\bibitem{buhrman}
H.~Buhrman and S.~Massar.
\newblock Causality and {Tsirel'son} bounds.
\newblock {\em Physical Review A, 72(5), 052103}, 2005.

\bibitem{buhrman2011}
H.~Buhrman, O.~Regev, G.~Scarpa, and R.~de~Wolf.
\newblock Near-optimal and explicit {Bell} inequality violations.
\newblock In {\em IEEE 26th Annual Conference on Computational Complexity
  (CCC)}, pages 157--166, 2011.

\bibitem{chsh_paper}
J.~F. Clauser, M.~A. Horne, A.~Shimony, and R.~A. Holt.
\newblock Proposed experiment to test local hidden-variable theories.
\newblock {\em Physical Review Letters}, 23(15):880--884, 1969.

\bibitem{cover2012}
T.~M. Cover and J.~A. Thomas.
\newblock {\em Elements of information theory}.
\newblock John Wiley \& Sons, 2012.

\bibitem{vandam}
W.~{\noopsort{Dam}{van Dam}}.
\newblock Implausible consequences of superstrong nonlocality.
\newblock {\em Natural Computing}, 12(1), 2013.

\bibitem{dinur2014}
I.~Dinur and D.~Steurer.
\newblock Analytical approach to parallel repetition.
\newblock In {\em Proceedings of the 46th Annual ACM Symposium on Theory of
  Computing}, STOC, pages 624--633, 2014.

\bibitem{ekert91}
A.~K. Ekert.
\newblock Quantum cryptography based on {Bell}'s theorem.
\newblock {\em Physical Review Letters}, 67(6):661--663, 1991.

\bibitem{manyquestions_fewanswers}
N.~Gisin.
\newblock {Bell} inequalities: Many questions, a few answers.
\newblock {\em The Western Ontario Series in Philosophy of Science}, 73, 2009.

\bibitem{ji2008}
S.-W. Ji, J.~Lee, J.~Lim, K.~Nagata, and H.-W. Lee.
\newblock Multisetting {Bell} inequality for qudits.
\newblock {\em Physical Review A}, 78(5):052103, 2008.

\bibitem{Jones_2010}
T.~G. Jones.
\newblock Explicit incidence bounds over general finite fields.
\newblock {\em arXiv preprint arXiv:1009.3899}, 2010.

\bibitem{Jones_thesis}
T.~G. Jones.
\newblock New quantitative estimates on the incidence geometry and growth of
  finite sets.
\newblock {\em PhD Thesis, University of Bristol. arXiv:1301.4853}, 2013.

\bibitem{junge2011}
M.~Junge and C.~Palazuelos.
\newblock Large violation of {Bell} inequalities with low entanglement.
\newblock {\em Communications in Mathematical Physics}, 306(3):695--746, 2011.

\bibitem{kasher2012}
R.~Kasher and J.~Kempe.
\newblock Two-source extractors secure against quantum adversaries.
\newblock {\em Theory of Computing}, 8(21), 2012.

\bibitem{kempe2008}
J.~Kempe, O.~Regev, and B.~Toner.
\newblock Unique games with entangled provers are easy.
\newblock In {\em 49th Annual IEEE Symposium on Foundations of Computer Science
  (FOCS)}. IEEE, 2008.

\bibitem{liang2009}
Y.-C. Liang, C.-W. Lim, and D.-L. Deng.
\newblock Reexamination of a multisetting {Bell} inequality for qudits.
\newblock {\em Physical Review A}, 80(5), 2009.

\bibitem{nonlocal_comp}
N.~Linden, S.~Popescu, A.~J. Short, and A.~Winter.
\newblock Quantum nonlocality and beyond: limits from nonlocal computation.
\newblock {\em Physical Review Letters}, 99(18):180502, 2007.

\bibitem{navas2008}
M.~Navascu{\'e}s, S.~Pironio, and A.~Ac{\'\i}n.
\newblock A convergent hierarchy of semidefinite programs characterizing the
  set of quantum correlations.
\newblock {\em New Journal of Physics}, 10(7), 2008.

\bibitem{IC_main}
M.~Paw{\l}owski, T.~Paterek, D.~Kaszlikowski, V.~Scarani, A.~Winter, and
  M.~Zukowski.
\newblock Information causality as a physical principle.
\newblock {\em Nature}, 461, 2009.

\bibitem{hyperbits}
M.~Paw{\l}owski and A.~Winter.
\newblock Hyperbits: The information quasiparticles.
\newblock {\em Physical Review A}, 85(2), 2012.

\bibitem{PR_box}
S.~Popescu and D.~Rohrlich.
\newblock Quantum nonlocality as an axiom.
\newblock {\em Foundations of Physics}, 24(3):379--385, 1994.

\bibitem{Rao2007}
A.~Rao.
\newblock An exposition of {Bourgain}Õs 2-source extractor.
\newblock {\em Electronic Colloquium on Computational Complexity (ECCC)},
  14(034), 2007.

\bibitem{raz2011}
R.~Raz.
\newblock A counterexample to strong parallel repetition.
\newblock {\em SIAM Journal on Computing}, 40(3):771--777, 2011.

\bibitem{RUV}
B.~W. Reichardt, F.~Unger, and U.~Vazirani.
\newblock A classical leash for a quantum system: command of quantum systems
  via rigidity of \textrm{CHSH} games.
\newblock In {\em Proceedings of the 4th conference on Innovations in
  Theoretical Computer Science (ITCS)}, 2013.

\bibitem{ST1983}
E.~Szemer{\'e}di and W.~T. Trotter~Jr.
\newblock Extremal problems in discrete geometry.
\newblock {\em Combinatorica}, 3(3-4):381--392, 1983.

\bibitem{boris_1}
B.~S. Tsirel'son.
\newblock Quantum generalizations of {Bell's} inequality.
\newblock {\em Letters in Mathematical Physics}, 4(2), 1980.

\bibitem{boris_2}
B.~S. Tsirel'son.
\newblock Quantum analogues of the {Bell} inequalities. the case of two
  spatially separated domains.
\newblock {\em Journal of Soviet Mathematics}, 36(4), 1987.

\bibitem{wangfunc}
G.~Wang.
\newblock Functional boxes, communication complexity and information causality.
\newblock {\em arXiv preprint arXiv:1109.4988}, 2011.

\bibitem{werner2001bell}
R.~F. Werner and M.~M. Wolf.
\newblock {Bell} inequalities and entanglement.
\newblock {\em Quantum Information and Computation (QIC)}, (3), 2001.

\end{thebibliography}

\appendix

\section[App.~Info.~theory background]{Information theory background} \label{appendix:info_theory}
Here we give a brief account of some definition from Information Theory. We do not require any knowledge of quantum information theory as the proofs in Appendix \ref{IC_app} are all based on classical information theory. The only place where we use some quantum ideas (which we must as Theorem \ref{IC_st} is formulated  in the quantum setting), is in Theorem \ref{IC_int} which we borrow directly from Paw{\l}owski and Winter \cite{hyperbits}.

For a random variables $X$ over a domain $\mathcal X$, we define the entropy of $X$ as
\[ H(X)= \sum_{x\in \mathcal X} \Pr[X=x]  \cdot \log\left(\frac{1}{\Pr[X=x]}\right).\]
The entropy of a random variable $X$ conditioned on the value of random variable $Y$ is the average (according to $Y$) entropy of random variables $X| Y=y$, i.e.
\[ H(X|Y)= \sum_{y\in \mathcal Y} \Pr[Y=y] \cdot H(X|Y=y),\]
where $H(X|Y=y)= \sum_{x\in \mathcal X} \Pr[X=x|Y=y]\cdot \log\left(\frac{1}{\Pr[X=x|Y=y]}\right)$.

\begin{remark}
The above definition is equivalent with the more concise definition of $H(X|Y)=H(XY)-H(Y)$.
\end{remark}

Now we are ready to define and state the main property of \textsl{mutual information}.

\begin{definition}[Mutual Information] Given random variables $X$ and $Y$ we define their mutual information via
\[ I(X;Y)= H(X)+H(Y)-H(XY)= H(X)-H(X|Y)=H(Y)-H(Y|X).\]
\end{definition}
The main fact that we shall need regarding the mutual information is the following.
\begin{proposition}[Chain Rule]
For any random variables $X,Y,Z,W$ we have
\[ I(X;YZ|W)= I(X;Y|W)+ I(X;Z|Y,W).\]
\end{proposition}
For a more detailed introduction to information theory as well as the proof of the above claims, we recommend the early chapters of \cite{cover2012}.

\section[App.~ Pairwise independent IC]{Pairwise independent Information Causality}\label{IC_app}
In this section we prove Theorem \ref{IC_st} resolving the main open problem of \cite{hyperbits}. Our starting point is the original theorem of Paw{\l}owski and Winter which is the binary version of our result. 
 \begin{theorem}[PW]\label{IC_int}
 Assume in the information causality game of Definition \ref{IC_def}, Alice's input $\mathbf{X}=(X_1,X_2,\ldots, X_n)$ is drawn from an unbiased pairwise independent distribution. Assume Alice's message to Bob $\alpha$ and output of Bob $Z$ have both only two outcomes. We have
 \[ \sum_{i=0}^N I(X_i; Z |b=i)=O(1) \, .\]
 \end{theorem}
 
The main idea for proving this theorem is to reduce the sizes of alphabets of Bob's output and Alice's message (to Bob)---while keeping a control of the information theoretic quantities of interest. We do this in two steps: first, we relax the assumption on the alphabet size of Bob's output $Z$, and then we relax the assumption on the alphabet size of the message $\alpha$. 
\begin{proposition}\label{IC_prop} Assume in the information causality game of Definition \ref{IC_def}, Alice's input $\mathbf{X}=(X_1,X_2,\ldots, X_n)$ is drawn from an unbiased pairwise independent distribution.  Assume Alice's message to Bob $\alpha$ has only two outcomes. Let $\Lambda$ be the alphabet  for Bob's output $Z$. Then,
\[ \sum_{i=0}^N I(X_i; Z |b=i)=O( |\Lambda|). \]
 \end{proposition}
 Note that what makes this alphabet reduction not so immediate is the fact that $I(X; Z_1 Z_2 \ldots Z_k)\geq m$ does not imply a lower bound on $\max_{i} \{ I(X; Z_i) \}$. The idea of alphabet reduction is as follows: We imagine Bob, after doing his measurement and getting some output $Z\in \Lambda$, performs a post-processing on $Z$ to produce a binary output $Z_{bin}$. This $Z_{bin}$ will be only a function of $Z$ and his input $b$, and will satisfy
 \[ I(X_i; Z_{bin} |b=i) \geq \frac{I(X_i; Z |b=i)}{|\Lambda|} .\]
This establishes the desired result since $\sum_{i=1}^n I(X_i; Z_{bin} |b=i)$ is at most a constant by Theorem \ref{IC_int}. In order, to define $Z_{bin}$ as a function of $Z$ and $b$ to satisfy the above property, we need the following lemma. 
 
 \begin{lemma} Let $X$ and $Y$ be two discrete random variables. Let $Y$ take values in some domain $\mathcal B$ and assume $X$ to be uniform on its domain $\mathcal A$. Then there exists a function $f:\mathcal B \rightarrow \{0,1\}$ such that
 \[ I(X; f(Y)) \geq \frac{ I(X;Y)}{|\mathcal B|} \, . \]
 
 \end{lemma}
 \begin{proof}
 Let $s$ be the ratio of the mutual information between $X$ and $Y$, and the entropy of $X$.  
 \[ I(X;Y)= sH(X)\, \Rightarrow \, H(X|Y)= (1-s) H(X)\, .\] 
 For any $j\in \mathcal B$, define $r_j= \Pr[Y=j]$ and   $t_j= \frac{H(X|Y=j)}{ H(X)}$.
Notice that since $\sum_{j\in \mathcal B} r_j=1$ we have
 \[
H(X | Y)= (1-s)H(X)= \sum_{j\in \mathcal B} r_j t_j H(X)= H(X) - \sum_{j\in \mathcal B}(1- t_j) r_j H(X) \, .
\]
Hence, there exists $j^*\in \mathcal B$ such that $r_{j^*} (1-t_{j^*}) \geq \frac{s}{|\mathcal B|}$. Define the function $f:\mathcal B\rightarrow \{0,1\}$ to be everywhere zero expect at $j^*$. We have
\[
 H(X| f(Y) )= \Pr[f(Y)=0] \:H(X | Y\neq j^*) + \Pr[f(Y)=1] \: H(X| Y=j^*) \,. 
 \]
 Since $X$ is unbiased over its domain, the first term is bounded by $(1-r_{j^*}) H(X)$. The second term is equal to $ r_{j^*} t_{j^*} H(X)$. Adding these two terms and using the choice of $j^*$ it follows that
\[ 
H(X|f(Y))\leq H(X)- r_{j^*}(1- t_j^*) H(X)\leq \,  H(X) \left(1- \frac{s}{|\mathcal B| } \right) \, , 
 \]
 which using the definition mutual information is seen to be the desired result.
 \end{proof}
 To prove Proposition \ref{IC_prop}, for $i \in \Lambda$ we let $f_i: \Lambda \rightarrow \{0,1\}$ be given by the above lemma applied to $X_i$ and $Z$ which is applicable since $X_i$ is unbiased. We let $Z_{bin}= f_i(Z)$. By Theorem \ref{IC_int} it follows that 
 \[ \sum_{i=1}^n I(X_i ;  Z_{bin}| b=i) = O(1)\, . \]
 On the other hand $I(X_i; Z |b=i) \leq I(X_i; Z_{bin}|b=i) |\Lambda |$ by the above lemma. Hence we are done with Proposition \ref{IC_prop}. Finally, we prove Theorem \ref{IC_st} by a reduction to the above proposition. 

Consider any protocol $\mathcal C$ for the information causality game in the setting of Theorem \ref{main_thm}. From $\mathcal C$, we produce an auxiliary protocol $\mathcal C^*$ with only a binary message (from Alice to Bob). By applying Proposition \ref{IC_prop} to $\mathcal C^*$ we prove the result. In $\mathcal C^*$ the binary message sent from Alice to Bob is denoted $\alpha_{bin}$, and  Bob's output is denoted by $\widetilde{Z}$. $\mathcal C^*$ is defined as follows:  Before even receiving their inputs,  Alice and Bob use shared randomness to ``guess" $\alpha$ which Alice was going to send to Bob in the original protocol $\mathcal C$. Let that guess be $\overline{\alpha}\in\Sigma$. Bob then proceeds with the original protocol treating $\overline{\alpha}$ as the message that Alice intended to send him in the original protocol $\mathcal C$.  By following the original protocol, he produces a $\overline{Z}\in \Lambda$. Alice given her input makes the appropriate measurement on her system according to the recipe given in $\mathcal C$ producing some $\alpha\in \Sigma$. She checks if $\alpha=\overline{\alpha}$ in which case she sends $\alpha_{bin}=1$ to Bob; otherwise she sends $\alpha_{bin}=0$. If Bob receives $\alpha_{bin}=1$, he sets $\widetilde{Z}=\overline{Z}$, and outputs $\widetilde{Z}\alpha_{bin}$. On the other hand, if he receives $\alpha_{bin}=0$, he discards $\overline{Z}$, and produces a completely random $\widetilde{Z}$. He again outputs $\widetilde{Z} \alpha_{bin}$. The choice that Bob outputs $\alpha_{bin}$ as well as $\widetilde{Z}$ is for technical convenience. In any case, this only increases the output size of Bob in $\mathcal C^*$ from $|\Lambda|$ to $2 |\Lambda|$, which only slightly affects the constant in $O(\cdot)$ in Theorem \ref{IC_st}. Now let us fill the rest of the details.
\begin{proof}[Proof of Theorem \ref{IC_st}]
The setup as usual is the following: Alice's input $\mathbf{X}=(X_1,X_2,\ldots, X_n)$ is taken from an unbiased pairwise independent distribution. Given her input, Alice makes a measurement on her system and sends $\alpha$ to Bob where $\alpha\in \Sigma$. Bob, given his input $b=i\in [N]$ and $\alpha$, outputs a $Z\in \Lambda$. Given any protocol $\mathcal C$ for the above, consider $\mathcal C^*$ which reduces the communication to binary according to the recipe described above: so the message would be $\alpha_{bin}$ and Bob's output $\widetilde{Z}\alpha_{bin}$.

Now by the chain rule, 
\[ I(X_i;\alpha_{bin}\widetilde{Z}|b=i)= I(X_i;\alpha_{bin}|b=i) + I(X_i;\widetilde{Z}|b=i, \alpha_{bin})\,. \]
So we have
\[ I(X_i;\alpha_{bin}\widetilde{Z}|b=i) \geq I(X_i;\widetilde{Z}|b=i, \alpha_{bin})\,. \]
We know $\sum_{i=1}^n I(X_i;\alpha_{bin}\widetilde{Z}|b=i)=O(|\Lambda|)$ by Proposition \ref{IC_prop}. Because in this situation Alice only communicates a bit $\alpha_{bin}$ to Bob who outputs $\widetilde{Z}\alpha_{bin}$ which has alphabet size $2|\Lambda|$. 

Now we want to analyze the term $ I(X_i;\widetilde{Z}|b=i, \alpha_{bin})$. First notice that the event $b=i$ is independent of the random variable $\alpha_{bin}$ because $\alpha_{bin}$ is produced entirely on Alice's side. Hence, 
\[\Pr[\alpha_{bin}=1  ] = \Pr[\alpha_{bin}=1|b=i] \]  
Now conditioned on $\alpha_{bin}=0$, $\widetilde{Z}$ is chosen completely at random. Hence,
\[I(X_i; \widetilde{Z}| b=i,\alpha_{bin}=0)=0\, .\]
This proves that 
\[ I(X_i;\widetilde{Z}|b=i, \alpha_{bin})= \Pr[\alpha_{bin}=1] I(X_i; \widetilde{Z}| b=i,\alpha_{bin}=1) \,. \]
Conditioning  on $\alpha_{bin}=1$ is independent of $X_i$ as one can imagine that Alice first produces $\alpha$ using the original protocol, and then selects $\overline{\alpha}\in \Sigma$ uniformly at random. Since $\overline{\alpha}$ is completely independent of everything else we see that the event $\alpha_{bin}=1$ is independent of $X_i$. The only effect of conditioning on the event $\alpha_{bin}=1$ is that the joint distribution of $(X_i,\widetilde{Z})$ becomes the same as the joint distribution of  $(X_i,Z)$. This means that this conditioning reduces us to the case of original protocol $\mathcal C$ where $\alpha\in \Sigma$ was communicated in its entirety. Hence, we have
\[ I(X_i; \widetilde{Z}| b=i,\alpha_{bin}=1)=I(X_i;Z| b=i). \]
 Noticing that $\Pr[\alpha_{bin}=1]=\frac{1}{|\Sigma|}$ and using
 \[\sum_{i=1}^n I(X_i;\alpha_{bin}\widetilde{Z}|b=i)=O(|\Lambda|)\, ,\] 
 we get our desired result
 \[ \sum_{i=0}^N I(X_i; Z |b=i)=O( |\Sigma| |\Lambda|)\,. \]
 \end{proof}

\section[App.~Distributed nonlocal computation]{Reduction to distributed nonlocal computation}\label{app:nonlocal}
Here we give an alternative proof of Theorem \ref{main_thm} by establishing a higher alphabet variant of powerful result of Linden et al. \cite{nonlocal_comp}. The main result of their work is that for a certain broad class of games, quantum and classical strategies are equivalent in their power. This class of games, called distributed non-local computation games, is defined as follows. 

\begin{definition}
Let $f:\F_2^n\rightarrow \F_2$ be any function and $\mathcal D$ any distribution on $\F_2^n$. In the $2$-party non-local computation of $(f,\mathcal D)$,  Alice receives $x\in \F_2^n$ uniformly at random and Bob receives $y\overset{\rm{def}}{=}z-x$ where $z$ is chosen according to $\mathcal D$. Alice and Bob succeed if their outputs $\alpha,\beta\in\F_2$ satisfy $\alpha+\beta= f(z)$. 
\end{definition}
\begin{theorem}[Linden et al.]\label{nonlocal_bin}
For any binary distributed computation problem $S$, given by $(f,\mathcal D)$, the entangled value and the classical value of the game coincide to the best linear approximation of $f$.
\[ \omega^*(S)=\omega(S)= \max_{l \in \mathcal L} \, \Pr_{z\in \mathcal D} [f(z)= l(z)]  \, ,\]
where $\mathcal L=\{ l:\F_2^n\rightarrow \F_2 \,|\, l(x+y)=l(x)+l(y)\}$ is the set of all linear functions. 
\end{theorem}
Now we can naturally define the distributed version of the $\bf\mathsf{CHSH_q}$ game following the above recipe. 

\begin{definition}\label{def:distributed_CHSH_q}
In $\bf\mathsf{CHSH_q^{dist}}$, Alice and Bob receive $(\alpha,\gamma)\in \F_q^2$ and $(\beta,\delta)\in \F_q^2$. Their objective is to produce outputs $a$ and $b$ satisfying $a+b=(\alpha+\beta)(\gamma+\delta)$. 
\end{definition}

Any strategy  $\mathcal P$ for $\bf\mathsf{CHSH_q}$ results in a natural strategy $\mathcal P^{\rm{dist}}$ for $\bf\mathsf{CHSH_q^{dist}}$ as follows: assume Alice and Bob have a strategy $\mathcal P$ succeeding in $\bf \mathsf{CHSH_q}$ game with probability $p_{win}=1/q+ E \, (q-1)/q $. In $\mathcal P^{\rm{dist}}$, first Alice and Bob use $\mathcal P$ on inputs $(\alpha,\delta)$ to produce $a_1$ and $b_1$, and then use $\mathcal P$ for a second time to produce $a_2$ and $b_2$ for inputs $(\gamma,\beta)$. Finally Alice outputs $a= a_1+a_2+\alpha\gamma$, and Bob outputs $b=b_1+b_2+\beta\delta$. 

We have the following proposition about $\mathcal P^{\rm{dist}}$ which is straightforward to establish.
\begin{proposition}\label{nonlocal_prop}
Let $\mathcal P$ be a regular protocol for $\bf \mathsf{CHSH_q}$ with bias $E$. Let $\mathcal P^{dist}$ be the resulting distributed protocol obtained through above procedure from $\mathcal P$. The winning probability of Alice and Bob in the protocol $\mathcal P^{\rm{dist}}$ has bias $E^2$, i.e. we have
\[ p_{win}=\frac{1}{q}+\frac{q-1}{q}E^2\, . \]
Moreover, the resulting protocol for $\bf \mathsf{CHSH_q^{dist}}$ is itself regular. 
\end{proposition}
The fact that $\mathcal P^{\rm{dist}}$ is regular when  $\mathcal P$ is regular because of the symmetry all elements of $\F_q^*$. To see the statement about the winning probability of $\mathcal P^{\rm{dist}}$ , notice that $\mathcal {P^*}$ succeeds if and only if the type of errors produced in two uses of $\mathcal P$ in $\mathcal P^{\rm{dist}}$ give errors  of opposite sign. Hence, the following calculation confirms the above proposition. 
\[ \left(\frac{1}{q}+\frac{q-1}{q}E\right)^2 +\sum_{k=1}^{q} \left(\frac{1}{q}-\frac{E}{q}\right)^2 =\frac{1}{q}+\frac{q-1}{q}E^2 \, . \]
Because of Proposition \ref{nonlocal_prop}, it suffices to prove an upper bound  of $2/q-1/q^2$ for the winning probability of regular strategies in $\bf\mathsf{CHSH_q^{dist}}$ to finish the proof of Theorem \ref{main_thm}. Notice that, as hinted before, this value is exactly the winning probability achieved by the trivial strategy in which both players just output $0$. This can be seen to be the best linear approximation to the polynomial $f(x,y)=xy$ which is aligned with the philosophy of Linden et al. \cite{nonlocal_comp} on equivalence of quantum and classical players for the distributed version of the game. 

Now we proceed to the proof of Theorem \ref{main_thm}. For simplicity, throughout this proof we shall assume $q$ is a prime. However, essentially the same argument works in general with some small modifications  mentioned at the end of the section. First we need the following lemma.
\begin{lemma}\label{fourier_prop}
Let $\omega$ be a $q^{\rm{th}}$ primitive root of unity. For any set of unit vector $u_x,v_y\in \C^n$ we have
 \begin{equation}\label{eqn:fourier_prop} \left|\sum_{x,y} \omega^{-xy} \langle u_x,v_y\rangle\right| \leq q^{3/2} \, .\end{equation}
 \end{lemma}
 \begin{proof}
Since the discrete Fourier transform matrix $H_{x,y}=\frac{1}{\sqrt{q}}\omega^{-xy}$ is unitary it follows that for any $f:\F_q\rightarrow \C$ and $g:\F_q\rightarrow \C$ then 
\[ \Big |\sum_{x,y} \omega^{-xy} f(x)g(y)\Big | \leq \sqrt{q} \|f\|_2 \|g\|_2\, . \]
Here we used Cauchy-Schwarz inequality between $\hat{f}(y)$ and $g(y)$. To finish the proof of the lemma, we apply the above fact to the $n$ coordinate functions $f_i,g_i:\F_q\rightarrow \C$ defined by $f_i(x)=\overline{u}_x(i)$ and $g_i(x)=v_y(i)$. Combining that with another Cauchy-Schwarz we get
\begin{align*}
\Big|\sum_{x,y} \omega^{-xy} \langle u_x,v_y\rangle \Big |&= \Big |\sum_{i,x,y} \omega^{-xy} f_i(x)g_i(y)\Big | \\
&\leq \sqrt{q} \sum_{i=1}^n  \left(\sum_{x\in \F_q} |u_x(i)|^2 \right)^{\frac{1}{2}}\left(\sum_{y\in \F_q} |v_y(i)|^2 \right)^{\frac{1}{2}} \\
& \leq \sqrt{q} \left(\sum_{i,x} |u_x(i)|^2\right)^{\frac{1}{2}} \left(\sum_{i,y} |v_y(i)|^2\right)^{\frac{1}{2}}\\
&= q^{3/2} \, .
\end{align*}
\end{proof}

We also need the following convenient notation for quantifying the probability of different types of error.  
\begin{notation} We let $p_k= \Pr_{ a,b\leftarrow \mathcal P^{\rm{dist}}}[ a+b= k+(\alpha+\beta)(\gamma+\delta)]$ according to $\mathcal P^{\rm{dist}}$ (defined after definition \ref{def:distributed_CHSH_q}).
\end{notation}
\begin{proof}[Proof of Theorem \ref{main_thm}]

Let $\omega$ be a primitive $q^{\rm{th}}$ root of unity. First notice that since $p_0= \frac{1}{q}+ \frac{q-1}{q} E^2$ and  $p_k= \frac{1}{q}- \frac{1}{q} E^2$ for $k\in \F_q^*$, it follows that 
 \begin{equation}\label{eqn:sec_nonlocal} \sum_{k=0}^{q-1} p_k\, \omega^k = E^2 + \left(\frac{1}{q}-\frac{E^2}{q}\right) \sum_{k=0}^{q-1} \omega^k= E^2 . \end{equation}
Denote by $x=(\alpha,\gamma)$ the input of Alice, and by $y=(\beta,\delta)$ the input of Bob. Let $P_x^a$ and $Q_y^b$ be the measurement operators for Alice and Bob given inputs $x$ and $y$. Define $U_x=\sum_{a=0}^{q-1} \omega^a P_x^a$ and $V_y=\sum_{b=0}^{q-1} \omega^b Q_y^b$. Since we have not assumed any bound on the dimension of Alice and Bob's system, we can assume the operators $P_x^a$ and $Q_y^b$ are projections. This implies that the operators $U_x$ and $V_y$'s are unitary. Now we have
\begin{align*}
 E^2= \sum_{k=0}^{q-1} p_k \omega^k & =\sum_{k=0}^{q-1} \,  \Ex_{x,y}\left[ \sum_{a+b=k+xy}  \langle \psi | P_x^a \otimes Q_y^b |\psi\rangle\right] \, \omega^k\\
 &=  \langle\psi| |\Ex_{x,y} \left[ \sum_{a,b=0}^{q-1} \omega^{a+b-xy} P_x^a\otimes Q_y^b\ \right] |\psi\rangle \\
 &= \frac{1}{q^2}\sum_{x,y\in F_q^2} \omega^{-xy} \langle \psi | U_x\otimes V_y |\psi\rangle\\
 &=\frac{1}{q^2}\sum_{x,y\in F_q^2} \omega^{-xy} \langle u_x|v_y\rangle,
  \end{align*}
 where $U_x^\dagger\otimes I |\psi\rangle= |u_x\rangle$ and $I\otimes V_y|\psi\rangle=|v_y\rangle$ are unit vectors.  So Lemma \ref{fourier_prop} implies $E\leq \frac{1}{\sqrt{q}}$ which establishes the desired result. 
 \end{proof} 
\paragraph{General prime powers.} In the above discussion we assumed $q$ was a prime. Here we present the slight modifications necessary in the more general case of $q=p^s$ with $s$ not necessarily $1$. 

Suppose we have a function $\chi:\F_q\rightarrow \C$ with the following properties:
\begin{enumerate}
\item For all $x\in \F_q$ we have $|\chi(x)|=1$.
\item $\chi(x+y)=\chi(x)\chi(y)$. 
\item $\sum_{x\in \F_q} \chi(x)=0$.
\end{enumerate}
Now if we replace the function $x\mapsto \omega^x$ with $x \mapsto \chi(x)$ in our argument we claim that the argument goes through exactly the same as before. To check this, first note that the $q\times q$ matrix $M_{xy}=\frac{1}{\sqrt{q}}\chi(xy)$ is again unitary. This is because for any $y\neq y'\in \F_q^*$ we have
\[ \sum_{x\in \F_q} \widebar{\chi(xy)}\cdot  \chi(xy')=\sum_{x\in \F_q} \chi(x(y'-y))=0, \]
where we used the second and third properties of $\chi$. This establishes that the analogue of equation (\ref{eqn:fourier_prop}) holds. 

Next note that the main property we used in equation (\ref{eqn:sec_nonlocal}) was the fact that $\sum_{k\in \F_q} \chi(k)=0$ which again holds here. Similarly, from the fact that $|\chi(a)|=1$ it follows that operators of the form  $U_x=\sum_{a\in \F_q} \chi(a) P_x^a$ are again unitary. It is easy to check the rest of our calculations are similarly valid given properties 1-3 of $\chi$.

Finally, we shall note that a function such as $\chi$ is easy to construct using an additive isomorphism between $\F_q$ and $\F_p^s$ where $p$ is the characteristic of $q=p^s$. \footnote{A particularly common choice is $\omega^{\Tr(\cdot)}$ where $\omega$ is $p^{\rm th}$ root of unity and $\Tr(\cdot)$ denotes the trace function given by $\alpha \mapsto \alpha+\alpha^p +\ldots +\alpha^{p^{s-1}}$ which can be shown to be a map from $\F_q$ to $\F_p$.} 

\section[App.~Projective transforms and applications]{Projective transforms and the classical value of $\bf\mathsf{CHSH_q}$}\label{lowerbound_classical}
The goal of this section is to give a proof of Lemma \ref{classical:lemma}. To do so we need to introduce some basic concepts from projective geometry.  The major advantages of working over the projective plane as opposed to the affine plane for our purposes  here  is that firstly, the duality between points and lines becomes quite transparent and clean over the projective plane, and secondly (and more importantly), it turns out that set of \emph{projective transformations}, while still preserving the point-line incidence structure, is much larger and richer than the group of affine transformations.

Recall that the points of the projective plane $\mathbb{P}\F_q^2$ are given by the triples $(x:y:z)$, with at least  one coordinate non-zero, where we identify any two points $(x:y:z)$ and $(\lambda x:\lambda y:\lambda z)$ for $\lambda \in \F_q^*$. The lines in the projective plane are given by triples $(l:m:n)$ consisting of all point $(x:y:z)$ satisfying $lx+my+nz=0$. Hence, $\mathbb{P}\F_q^2$ consists of a total of $q^2+q+1$ points and lines, where each line contains $q+1$ points, and similarly each point is contained in $q+1$ lines. One can go from the projective
plane to the affine plane by  discarding all the points on the ``line at infinity", which consists 
 of points of the form $(a:b:0)$. Any point remaining can then be put in the form $(a:b:1)$ which corresponds to the point $(a,b)$
of the affine plane. Two points lying on the same a vertical line in the affine plane
are $(a:b:1)$ and $(a:b':1)$. These lie on the projective line $(1:0:-a)$.
These projective lines all go through the point $(0:1:0)$, or 
the ``vertical infinity".  

Next we describe the concept of a projective transformation. Any $3\times 3$ invertible matrix $A$ induces a map on $\mathbb{P}\F_q^2$ by its action on $\F_q^3$. Two matrices $A$ and $B$ induce the same action on $\mathbb{P}\F_q^2$, if they are related to by $A= \lambda B$ by some $\lambda\in \F_q^*$. Each element of this equivalence class represents a projective transformation on the projective plane. 

Now we are almost ready to prove Lemma \ref{classical:lemma}. The main idea of the proof is to start from the given sets of points $P$ and lines $L$ with large $I(P,L)$ which possibly are far from satisfying the conditions of Fact \ref{important_fact} and apply a projective transformation. After, this we discard some of our points and lines so that the hypothesis of Fact \ref{important_fact} is satisfied. We argue that a good portion of our incidences would remain after the above deletion process finishing the proof. The details are as follows.

\begin{proof}[Proof of Lemma \ref{classical:lemma}]
Let $P$ and $L$ in $\mathbb{P}\F_q^2$ be the set of points and lines given by the hypothesis. In the argument we shall operate under the assumption that $|P|,|L|\leq \frac{q}{2}$. 
Indeed, one can cut the sizes of $P$ and $L$ by any constant factor without losing much in  $|I(P,L)|$ as follows: first we shrink $|P|$ by keeping the $\frac{q}{2}$ points with the most number of incident lines---which shrinks $|I(P,L)|$ only by a factor of $\frac{2|P|}{q}=O(1)$ factor---and after keeping those lines with the most number incident points (among the $\frac{q}{2}$ size points that we kept).

Now we apply a random projective transformation to $\mathbb{P}\F_q^2$ which can be seen to be equivalent to the following operation: Take a line $l$ from all $q^2+q+1$ lines in $\mathbb{P}\F_q^2$ uniformly at random to be the new line at the infinity. Choose a random point on $l$ to be the new vertical point at infinity (and a random points from the remaining $q$ points on the line at the infinity to be the new horizontal point at infinity, but this latter point plays no role for us). Next we discard all but a single line from any subset of $L$ that are parallel after this operation. Similarly, we discard all but one point from any set of points in $P$ that are on the same vertical line  after this operation. We also discard any point of $P$ on the new line at infinity; this puts us back in the case of the affine plane $\F_q^2$. Let $P'$ and $L'$ be the set of points and lines in $\F_q^2$ obtained after the above operation. To finish the proof, it suffices to show that in expectation $P'$, $L'$ and $I(P',L')$ would remain within constant factors of their original sizes. 

Let us first compute the probability that two given lines $\ell_1$ and 
$\ell_2$ have the same slope.
If we choose a random line as $\ell_\infty$, two lines will have the 
same slope if they intersect the line $\ell_\infty$ at the same point. 
The probability of this event for a specific pair of lines is $1/(q+1)$
(assuming that neither is the line at infinity), 
as this happens exactly when $\ell_\infty$ intersects $\ell_1$ at the same
point that $\ell_2$ intersects $\ell_1$.  We now compute the probability
that a given line survives. The probability that it is not the line at 
infinity is $\frac{q^2+q}{q^2+q+1}$, and given that it is not the line 
at infinity, it is eliminated with probability at most $\frac{q-1}{2(q+1)}$,
since there are $\frac{q-1}{2}$ other lines that could eliminate it by
having the same slope.
This gives a total probability that it is eliminated of
$\frac{q(q+3)}{2(q^2+q+1)} < \frac{1}{2}$. This shows that the expected 
number of lines that survive is at least $\frac{q+1}{4}$. 

We next need to analyze the effect of choosing the point at vertical 
infinity at random and throwing out points lying on the same vertical line. The argument is analogous to the above. Again, the
probability that a given two points lie on a vertical is $1/(q+1)$. The
probability that a point is not on the line at infinity is
$\frac{q^2}{q^2+q+1}$. By the same argument as before, we have that the
expectation that a given point survives is at least
$\frac{q^2(q+3)}{2(q^2+q+1) (q+1)}$.
Thus, the expected number of points that survive is also $\Theta(q)$.
Furthermore the process of
throwing out points and the process of throwing out lines are independent. 
Thus, the probability that we keep a point-line incidence is at 
least $\Omega(1)$ which finishes the proof.
\end{proof}


 \end{document}